\documentclass[11pt]{article}
\usepackage{amsmath}
\usepackage{amsfonts}
\usepackage{amsthm}
\usepackage{amsmath,amssymb}
\usepackage{graphicx}
\usepackage{color}
\usepackage{url}
\usepackage{graphicx}
\usepackage{appendix}
\usepackage{booktabs}

\allowdisplaybreaks[4]

\newtheorem{theorem}{Theorem}
\newtheorem{lemma}{Lemma}
\newtheorem{corollary}{Corollary}
\newtheorem{property}{Property}
\newtheorem{definition}{Definition}
\newtheorem{proposition}{Proposition}

\newtheorem{example}{Example}

\newcommand{\RNum}[1]{\uppercase\expandafter{\romannumeral #1\relax}}

\newcommand{\gf}{{\mathbb{F}}}

\newcommand{\ls}[1]
    {\dimen0=\fontdimen6\the\font\lineskip=#1\dimen0
     \advance\lineskip.5\fontdimen5\the\font
     \advance\lineskip-\dimen0
     \lineskiplimit=0.9\lineskip
     \baselineskip=\lineskip
     \advance\baselineskip\dimen0
     \normallineskip\lineskip\normallineskiplimit\lineskiplimit
     \normalbaselineskip\baselineskip
     \ignorespaces}

\begin{document}

\bibliographystyle{abbrv}

\title{The Explicit values of the UBCT, the LBCT and the DBCT of the inverse function}

\author{
Yuying Man\footnotemark[1]\thanks{Y. Man and X. Zeng are with Hubei Key Laboratory of Applied Mathematics, Faculty of Mathematics and Statistics, Hubei University, Wuhan 430062, China. Email:yuying.man@aliyun.com, xzeng@hubu.edu.cn},
Nian Li\footnotemark[2]\thanks{N. Li and Z. Liu are with Hubei Key Laboratory of Applied Mathematics, School of Cyber Science and Technology, Hubei University, Wuhan 430062, China. Email: nian.li@hubu.edu.cn, liuzhen@hubu.edu.cn},
Zhen Liu \footnotemark[2],
Xiangyong Zeng \footnotemark[1]
}
\date{\today}
\maketitle

\thispagestyle{plain} \setcounter{page}{1}

\begin{abstract}

Substitution boxes (S-boxes) play a significant role in ensuring the resistance of block ciphers against various attacks. The Upper Boomerang Connectivity Table (UBCT), the Lower Boomerang Connectivity Table (LBCT) and the Double Boomerang Connectivity Table (DBCT) of a given S-box are crucial tools to analyze its security concerning specific attacks. However, there are currently no related results for this research.
The inverse function is crucial for constructing S-boxes of block ciphers with good cryptographic properties in symmetric cryptography. Therefore, extensive research has been conducted on the inverse function, exploring various properties related to standard attacks. Thanks to the recent advancements in boomerang cryptanalysis, particularly the introduction of concepts such as UBCT, LBCT, and DBCT, this paper aims to further investigate the properties of the inverse function $F(x)=x^{2^n-2}$ over $\gf_{2^n}$ for arbitrary $n$. As a consequence, by carrying out certain finer manipulations of solving specific equations over $\gf_{2^n}$, we give all entries of the UBCT, LBCT of $F(x)$ over $\gf_{2^n}$ for arbitrary $n$.
Besides, based on the results of the UBCT and LBCT for the inverse function, we determine that $F(x)$ is hard when $n$ is odd. Furthermore, we completely compute all entries of the DBCT of $F(x)$ over $\gf_{2^n}$ for arbitrary $n$. Additionally, we provide the precise number of elements with a given entry by means of the values of some Kloosterman sums. Further, we determine the double boomerang uniformity of $F(x)$ over $\gf_{2^n}$ for arbitrary $n$. Our in-depth analysis of the DBCT of $F(x)$ contributes to a better evaluation of the S-box's resistance against boomerang attacks.

\noindent{\bf Keywords} Upper Boomerang Connectivity Table, Lower Boomerang Connectivity Table, Double Boomerang Connectivity Table, double boomerang uniformity

\noindent{\bf MSC (2020)} 94A60, 11T06

\end{abstract}

\section{Introduction}

Substitution boxes (S-boxes) are essential components of symmetric cryptographic algorithms since S-boxes usually are the only nonlinear elements of these cryptosystems. The functions used to design S-boxes should have good cryptographic properties in order to resist various kinds of cryptanalytic attacks.

Differential attack, introduced by Biham and Shamir \cite{diff-at} in 1991, is one of the most fundamental cryptanalytic tools to assess the security of block ciphers. The main idea is to search for non-random pairs of input and output differences of the cipher with high probability.
The Difference Distribution Table (DDT) and the differential uniformity of S-boxes, introduced by Nyberg \cite{def-DDT} in 1993, can be used to measure the ability of an S-box to resist the differential attack. The smaller the differential uniformity $\delta_F$ of an S-box $F$, the stronger its ability to resist differential attack. For practical applications in cryptography, it is usually desirable to employ mappings with differential uniformity no greater than $4$. For example, the AES uses the inverse function $x\mapsto x^{-1}$ over $\mathbb{F}_{2^n}$, which has differential uniformity $4$ for even $n$ and $2$ for odd $n$.

Boomerang attack is another crucial cryptanalytical technique on block cyphers, introduced by Wagner \cite{BA} in 1999, which can be considered an extension of the classical differential attack. To analyze the boomerang attack of block cyphers in a better way, analogous to the DDT concerning the differential attack, in Eurocrypt 2018, Cid et al. in \cite{BCT} introduced a new tool known as Boomerang Connectivity Table (BCT) to measure the resistance of an S-box against boomerang attacks. The BCT tool allows us to easily evaluate the probability of the boomerang switches when it covers one round. Small entries in the BCT of a cipher prevent it from attacks related to the boomerang cryptanalysis.
When studying how to extend the BCT theory to boomerang switches on more rounds, Wang et al. \cite{BDT} proposed the concept of the Boomerang Difference Table (BDT), a variant of the BCT with one supplementary variable fixed. Moreover, BDT renamed as Upper Boomerang Connectivity Table (UBCT) in \cite{ULBCT}, its variant called BDT' is denoted by $\mathcal{D}_{BCT}$ in \cite{BDT'} and renamed as Lower Boomerang Connectivity Table (LBCT) in \cite{ULBCT}. The names of UBCT and LBCT emphasize better the fact that the Upper (resp. Lower) BCT focuses on the upper (resp. lower) characteristic.
These different tables enhance the analysis and evaluation of boomerang distinguishers. Under specific assumptions, the aforementioned tables enable the estimation of the probability of boomerang distinguishers, even when the middle part consists of multiple rounds \cite{IT-D23}.
It is worth noting that all the above tables are defined for one S-box, while in order to more efficiently calculate the probabilities of boomerang distinguishers and capture the behavior of two consecutive S-boxes in boomerang attacks, Hadipour et al. in \cite{def-DBCT} introduced a new tool called Double Boomerang Connectivity Table (DBCT) and used it for automatic searching for boomerang distinguishers.
For a boomerang distinguishers, its probability is closely related to the S-box being used. Specifically, the probability of the distinguisher is related to the DBCT, BCT and the DDT of the S-box. To demonstrate the effect of the S-box on the probability of the boomerang distinguisher, \cite[Table 4]{p-DBCT} compares the uniformity of DDT, BCT and DBCT for different S-boxes and gives the probabilities for the 7-round distinguisher under different S-boxes.
These observations indicate that, in addition to the uniformity of BCT and DDT, the uniformity of DBCT is a new measure criterion to evaluate the performance of S-box for resisting boomerang attacks. Therefore, the double boomerang uniformity should be used together with the boomerang uniformity to have a better evaluation of the S-box against the boomerang attack\cite{p-DBCT}.
The explicit values of the entries of DBCT and their cardinalities are crucial tool to test the resistance of block ciphers based on variants of the function against cryptanalytics such as differential and boomerang attacks. However, obtaining these entries and the cardinalities in DBCT for a given S-box is challenging, and therefore there are currently no related results for this research.

In this paper, we investigate the specific tables of the inverse function by carrying out some finer manipulations of solving certain equations over $\mathbb{F}_{2^n}$.
Firstly, using the techniques for solving equations over finite fields and the result of the DDT of the inverse function, we provide all entries of the UBCT and LBCT of the inverse function. Next, based on the results of the UBCT and LBCT of the inverse function, we determine that $F(x)$ is hard when $n$ is odd. Furthermore, utilizing the balance property of the trace function and the value of some Kloosterman sums, we compute all entries of the DBCT of the inverse function $F(x)=x^{2^n-2}$ over $\gf_{2^n}$. Additionally, we provide an accurate value of the number of elements with a given entry. Besides, we determine the double boomerang uniformity of $F(x)$ over $\gf_{2^n}$ for arbitrary $n$.

The remainder of this paper is organized as follows. In Section \ref{pre}, we present some basic notions and a few known helpful results in the technical part of the paper. Section \ref{DBCT-result} provides explicit values of all entries in the DBCT of $F(x)$ over $\gf_{2^n}$ for arbitrary $n$. Section \ref{con-remarks} concludes this paper.

\section{Preliminaries}\label{pre}

Throughout this paper, $\mathbb{F}_{2^n}$ denotes the finite field with $2^n$ elements  and  ${\rm Tr}_{1}^n(\cdot)$ denotes the  (absolute) trace function from $\mathbb{F}_{2^n}$ onto its prime field $\mathbb{F}_{2}$, where $n$ is a positive integer. Recall that  for $x\in\mathbb{F}_{2^n}$, ${\rm Tr}_{1}^n(x)=\sum_{i=0}^{n-1}x^{2^i}$.

In this section,  we recall some basic definitions and present some results which will be used frequently in this paper.

\begin{definition}\label{definition-DDT}\rm (\cite{def-DDT})
Let $F(x)$ be a  mapping from $\mathbb{F}_{2^n}$ to itself. The Difference Distribution Table (DDT) of $F(x)$ is a $2^n \times 2^n$ table where the entry at $(a,b)\in \mathbb{F}_{2^n}^2$ is defined by
$${\rm DDT}_{F}(a,b)=|\{x \in \gf_{2^n}: F(x)+F(x+a)=b \}|.$$

The mapping $F(x)$ is said to be \emph{differentially $\delta$-uniform} if $\delta(F)=\delta$ \cite{def-DDT}, and accordingly $\delta(F)$ is called the \emph{differential uniformity} of $F(x)$, where
$$
\delta(F)=\max _{a,b \in \mathbb{F}_{2^n},a\ne 0} {\rm DDT}_{F}(a,b).
$$
When $F$ is used as an S-box inside a cryptosystem, the smaller the value $\delta(F)$ is, the better the contribution of $F$ to the resistance against differential attack.


\end{definition}

\begin{definition}\label{definition-BCT}\rm (\cite{BCT})
Let $F: \mathbb{F}_2^n\rightarrow \mathbb{F}_2^n$ be a permutation over $\mathbb{F}_{2^n}$.
The boomerang Connectivity Table (BCT) is defined by the  $2^n\times 2^n$ table defined for $a, b\in \mathbb{F}_{2^n}$ by
$$
{\rm BCT}_F(a,b)=|\{x\in \mathbb{F}_{2^n} : F^{-1}(F(x)+b)+F^{-1}(F(x+a)+b)=a\}|.
$$

The boomerang uniformity~(see. \cite{BCT-u}) of $F$ is defined by
$$
\beta(F)=\max_{a,b\in \mathbb{F}_{2^n},ab\neq 0}{\rm BCT}_F(a,b).
$$
\end{definition}

In \cite{BDT}, two tables known as BDT and BDT' were introduced. However, Delaune et al. renamed them to UBCT and LBCT in \cite{ULBCT}, respectively emphasizing the upper and lower characteristics. These tables exhibit symmetry, with properties in one table having corresponding counterparts in the other. Delaune et al. defined the UBCT and LBCT as follows.

\begin{definition}\label{definition-UBCT}\rm (\cite{ULBCT})
Let $F(x)$ be a mapping from $\mathbb{F}_{2^n}$ to itself. The Upper Boomerang Connectivity Table (UBCT) is a $2^n\times 2^n\times 2^n$ table defined for $(a,b,c)\in \mathbb{F}_{2^n}^3$ by
\begin{equation*}
\begin{aligned}
{\rm UBCT}_{F}(a,b,c)=|\{x \in \gf_{2^n}: &\,\, F^{-1}(F(x)+c)+F^{-1}(F(x+a)+c)=a,\\
&\,\, F(x)+F(x+a)=b \}|.
\end{aligned}
\end{equation*}
\end{definition}

\begin{definition}\label{definition-LBCT}\rm (\cite{ULBCT})
Let $F(x)$ be a mapping from $\mathbb{F}_{2^n}$ to itself. The Lower Boomerang Connectivity Table (LBCT) is a $2^n\times 2^n\times 2^n$ table defined for $(b,c,d)\in \mathbb{F}_{2^n}^3$ by
\begin{equation*}
\begin{aligned}
{\rm LBCT}_{F}(b,c,d)=|\{x \in \gf_{2^n}: &\,\,F^{-1}(F(x)+d)+F^{-1}(F(x+b)+d)=b,\\
&\,\,F(x)+F(x+c)=d \}|.
\end{aligned}
\end{equation*}
\end{definition}

In order to capture the properties of two continuous S-boxes in boomerang attacks, Yang et al. in \cite{p-DBCT} study the DBCT. It is worth noting that the concept of DBCT was initially introduced and algorithmically defined in \cite{def-DBCT}. In fact, the notation for DBCT used in \cite{p-DBCT} is the same as in \cite{def-DBCT}, but the authors in \cite{p-DBCT} adopted a more succinct definition. Yang et al. defined the DBCT in the following.

\begin{definition}\label{definition-DBCT}\rm (\cite{p-DBCT})
Let $F(x)$ be a mapping from $\mathbb{F}_{2^n}$ to itself. The Double Boomerang Connectivity Table (DBCT) is a $2^n\times 2^n$ table defined for $(a,d)\in \mathbb{F}_{2^n}^2$ by
$$
{\rm DBCT}_{F}(a,d)=\sum\limits_{a,d} {\rm dbct}(a,b,c,d),
$$
where ${\rm dbct}(a,b,c,d)={\rm UBCT}_{F}(a,b,c)\cdot {\rm LBCT}_{F}(b,c,d)$.

The double boomerang uniformity of $F(x)$ is the largest value in the DBCT except for the first row and the first column:
$$
\beta_d=\max_{a,d\in \mathbb{F}_{2^n}^*}{\rm DBCT}_F(a,d).
$$
\end{definition}

\begin{property}\label{pro-DBCT}\rm (\cite{p-DBCT})
Let $F(x)$ be a mapping from $\mathbb{F}_{2^n}$ to itself. For $a,b,c,d \in \gf_{2^n}^*$, it is note that ${\rm DBCT}_{F}(a,d)$ can be expressed as the sum of two parts:
$$
{\rm DBCT}_{F}(a,d)=\sum\limits_{b=c} {\rm dbct}(a,b,c,d)+\sum\limits_{b\ne c} {\rm dbct}(a,b,c,d),
$$
and nonzero ${\rm dbct}(a,b,c,d)$ occurs mainly when $b=c$. Consequently,
\begin{equation*}
\begin{aligned}
{\rm DBCT}_{F}(a,d)&=\sum\limits_{b, c} {\rm dbct}(a,b,c,d)\\
&=\sum\limits_{b, c} {\rm UBCT}_{F}(a,b,c)\cdot {\rm LBCT}_{F}(b,c,d)\\
&\geq \sum\limits_{b=c} {\rm UBCT}(a,b,c)\cdot  {\rm LBCT}(b,c,d)\\
&=\sum\limits_{b} {\rm DDT}_{F}(a,b)\cdot {\rm DDT}_{F}(b,d).
\end{aligned}
\end{equation*}
For $a=0$ or $d=0$, it can be easily obtain that
$$
{\rm DBCT}(a,d)=\sum\limits_{c}{\rm UBCT}_{F}(0,0,c)\cdot {\rm LBCT}_{F}(0,c,d)=2^{2n},
$$
$$
{\rm DBCT}(a,d)=\sum\limits_{b}{\rm UBCT}_{F}(a,b,0)\cdot {\rm LBCT}_{F}(b,0,0)=2^{2n}.
$$
\end{property}

In \cite{p-DBCT}, the authors noticed that certain S-boxes exhibit a remarkable property: the set of nonzero ${\rm dbct}(a,b,c,d)$ always satisfies $b=c$. This is an exciting property as it turns the "$\geq$" in Property \ref{pro-DBCT} into a more desirable "$=$". When this property applies to an S-box, it is defined as a hard S-box.

\begin{definition}\label{definition-Hard}\rm (\cite{p-DBCT})
Let $F(x)$ be a mapping from $\mathbb{F}_{2^n}$ to itself. For $a,d\in \gf_{2^n}^*$, $F(x)$ is hard if the following holds:
\begin{equation*}
\begin{aligned}
{\rm DBCT}_{F}(a,d)&=\sum\limits_{b, c} {\rm UBCT}_{F}(a,b,c)\cdot {\rm LBCT}_{F}(b,c,d)\\
&= \sum\limits_{b=c} {\rm UBCT}(a,b,c)\cdot  {\rm LBCT}(b,c,d)\\
&=\sum\limits_{b} {\rm DDT}_{F}(a,b)\cdot {\rm DDT}_{F}(b,d).
\end{aligned}
\end{equation*}
\end{definition}

To state our main result, we need the Kloosterman sum $K(1)$ in $\gf_{2^n}$, which we introduce below.
\begin{definition}\label{definition-K}\rm (\cite{Kloosterman})
The Kloosterman sum $K(1)$ in $\gf_{2^n}$ is defined by
$$
K(1)=\sum\limits_{x\in \gf_{2^n}} (-1)^{{\rm Tr_1^n}(x^{-1}+x)}.
$$
Here by convention we use $(-1)^{{\rm Tr_1^n(0^{-1})}}:=1$.

The explicit value of $K(1)$ can be obtained easily from a formula of Carlitz
$$
K(1)=1+\frac{(-1)^{n-1}}{2^{n-1}}\sum\limits_{i=0}^{\lfloor \frac{n}{2} \rfloor}(-1)^i \binom{n}{2i}7^i.
$$
\end{definition}

The following lemma will be used frequently in this paper.

\begin{lemma}\label{lemma1-root}\rm (\cite{def-squ})
Let $a, b, c \in \mathbb{F}_{2^n}$, $a\ne 0$ and $F(x)=ax^2+bx+c$. Then
\begin{itemize}
\item [\rm (i)] $F(x)$ has exactly one root in $\mathbb{F}_{2^n}$ if and only if $b=0$;
\item [\rm (ii)] $F(x)$ has exactly two roots in $\mathbb{F}_{2^n}$ if and only if $b\ne 0$ and ${\rm Tr}_{1}^n(\frac{ac}{b^2})=0$;
\item [\rm (iii)] $F(x)$ has no root in $\mathbb{F}_{2^n}$ if and only if $b\ne 0$ and ${\rm Tr}_{1}^n(\frac{ac}{b^2})=1$.
\end{itemize}
\end{lemma}

The known results on the DDT of the inverse function are given as below, which will be useful for proving our results.

\begin{lemma}\label{inverse-DDT}\rm (\cite{ex-FBCT})
Let $F(x)=x^{2^n-2}$ be the inverse function over $\gf_{2^n}$. Then
\begin{itemize}
\item [\rm (i)] if $a=0$ and $b=0$, then ${\rm DDT}_{F}(a,b)=2^n$;
\item [\rm (ii)] if $a=0$ and $b\ne 0$, then ${\rm DDT}_{F}(a,b)=0$;
\item [\rm (iii)] if $a\ne 0$ and $b=0$, then ${\rm DDT}_{F}(a,b)=0$;
\item [\rm (iv)] if $a\ne 0$ and $ab=1$, then
$$
{\rm DDT}_{F}(a,b)=
\begin{cases}
    4, &  {\rm if}\,\, {\rm Tr}_{1}^n(1)=0; \\
    2, &  {\rm if}\,\,  {\rm Tr}_{1}^n(1)=1; \\
\end{cases}
$$
\item [\rm (v)] if $a\ne 0$, $b\ne 0$, and $ab\ne 1$, then
$${\rm DDT}_{F}(a,b)=
\begin{cases}
    2, &  {\rm if}\,\, {\rm Tr}_{1}^n(\frac{1}{ab})=0; \\
     0, &  {\rm if}\,\,  {\rm Tr}_{1}^n(\frac{1}{ab})=1. \\
\end{cases}
$$
\end{itemize}
\end{lemma}

\section{Main result and its proof}\label{DBCT-result}
In this section, we study explicit values of all entries of the DBCT of the inverse function. According to the Definition \ref{definition-DBCT}, in order to determine the DBCT of the inverse function, we first need to study its UBCT and LBCT.

Now we begin to deal with the UBCT of the inverse function.
We note that Definition \ref{definition-UBCT} applies only to invertible functions, i.e., permutations. By the definitions of UBCT, computing the UBCT of a permutation $F(x)$ over $\gf_{2^n}$ requires knowledge of its compositional inverse $F^{-1}(x)$.
However, it is very difficult to compute the compositional inverse with explicit form of $F(x)$ over $\gf_{2^n}$. Therefore, it is interesting and meaningful to compute the UBCT of $F(x)$ without considering $F^{-1}(x)$.
In the following, we present an equivalent formula to compute UBCT without knowing $F^{-1}(x)$ and $F(x)$ simultaneously.

Let $a,c\in \gf_{2^n}^*$ and $F(x)$ be a permutation polynomial over $\gf_{2^n}$. Let $y=F^{-1}(F(x)+c)$ in the first equation $F^{-1}(F(x)+c)+F^{-1}(F(x+a)+c)=a$ in Definition \ref{definition-UBCT}. For one thing, we have $F(x)+F(y)=c$ from the assumption; for the other, we have $F(x+a)+F(y+a)=c$. Since $F(x)$ is bijective, for $a,c\in \gf_{2^n}^*$, ${\rm UBCT}_{F}(a,b,c)$ equals to the number of solutions of the following equation system
\begin{equation}\label{UBCT}
\begin{cases}
    F(x)+F(y)=c, \\
    F(x+a)+F(y+a)=c, \\
    F(x+a)+F(x)=b.  \\
\end{cases}
\end{equation}
According to the above discussion, it is clear from (\ref{UBCT}) that the condition with permutation property is not necessary.

The UBCT of the inverse function is given by the following lemma.

\begin{lemma}\label{inverse-UBCT}
Let $F(x)=x^{2^n-2}$ be the inverse function over $\mathbb{F}_{2^n}$. Then
\begin{itemize}
\item [\rm (i)] if $a=0$ and $b=0$, then ${\rm UBCT}_{F}(a,b,c)=2^n$;
\item [\rm (iv)] if $a\ne 0$, $b\ne 0$,  then
$${\rm UBCT}_{F}(a,b,c)=
\begin{cases}
    2, &  {\rm if}\,\, c=0\,\, {\rm and}\,\,ab=1\\
       &    {\rm or}\,\, c=0,\,\, ab\ne 1\,\, {\rm and}\,\, {\rm Tr}_{1}^n(\frac{1}{ab})=0 \\
       &  {\rm or}\,\, c\ne 0,\,\, b=c\,\, {\rm and}\,\, ab=1\\
       &    {\rm or}\,\, c\ne 0,\,\, b=c\,\, {\rm and}\,\, {\rm Tr}_{1}^n(\frac{1}{ac})=0;\\
     0, &  {\rm otherwise} \\
\end{cases}
$$
for $n$ is odd; and for $n$ is even,
$${\rm UBCT}_{F}(a,b,c)=
\begin{cases}
     4, &  {\rm if}\,\, c=0\,\, {\rm and}\,\,ab=1\\
        &  {\rm or}\,\, c\ne 0,\,\, b=c\,\, {\rm and}\,\, ab=1\\
        &    {\rm or}\,\, c\ne 0,\,\,b\ne c,\,\, ab=1\,\, {\rm and}\,\, ac^2+abc+b=0;\\
     2, &   {\rm if}\,\, c=0,\,\, ab\ne 1\,\, {\rm and}\,\, {\rm Tr}_{1}^n(\frac{1}{ab})=0 \\
        &  {\rm or}\,\, c\ne 0,\,\, b=c\,\,  {\rm and}\,\, {\rm Tr}_{1}^n(\frac{1}{ac})=0;\\
     0, &  {\rm otherwise}. \\
\end{cases}
$$
\end{itemize}
\end{lemma}

\begin{proof}
See Appendix A.
\end{proof}

Next we determine the LBCT of the inverse function.
Similar to the discussion at the beginning of Lemma \ref{inverse-UBCT}, we present an equivalent formula to compute LBCT of $F(x)$ in the following.

Let $y=F^{-1}(F(x)+d)$ in the first equation $F^{-1}(F(x)+d)+F^{-1}(F(x+b)+d)=b$ in Definition \ref{definition-LBCT}. For one thing, we have $F(x)+F(y)=d$ from the assumption; for the other, we have $F(x+b)+F(y+b)=d$. For the second equation $F(x)+F(x+c)=d$ in Definition \ref{definition-LBCT}, we get $y=x+c$.
Since $F(x)$ is bijective, for $c,d\in \gf_{2^n}^*$, ${\rm LBCT}_{F}(b,c,d)$ equals to the number of solutions of the following equation system
\begin{equation}\label{LBCT}
\begin{cases}
    F(x)+F(y)=d, \\
    F(x+b)+F(y+b)=d, \\
    y=x+c.  \\
\end{cases}
\end{equation}

The LBCT of the inverse function is given by the following lemma.

\begin{lemma}\label{inverse-LBCT}
Let $F(x)=x^{2^n-2}$ be the inverse function over $\mathbb{F}_{2^n}$. Then
\begin{itemize}
\item [\rm (i)] if $c=0$ and $d=0$, then ${\rm LBCT}_{F}(b,c,d)=2^n$;
\item [\rm (iv)] if $c\ne 0$, $d\ne 0$,  then
$${\rm LBCT}_{F}(b,c,d)=
\begin{cases}
    2, &  {\rm if}\,\, b=0\,\, {\rm and}\,\,cd=1\\
       &    {\rm or}\,\, b=0,\,\, cd\ne 1\,\, {\rm and}\,\, {\rm Tr}_{1}^n(\frac{1}{cd})=0 \\
       &  {\rm or}\,\, b\ne 0,\,\, b=c\,\, {\rm and}\,\, cd=1\\
       &    {\rm or}\,\, b\ne 0,\,\, b=c\,\, {\rm and}\,\, {\rm Tr}_{1}^n(\frac{1}{cd})=0;\\
     0, &  {\rm otherwise} \\
\end{cases}
$$
for $n$ is odd; and for $n$ is even,
$${\rm LBCT}_{F}(b,c,d)=
\begin{cases}
     4, &  {\rm if}\,\, b=0\,\, {\rm and}\,\,cd=1\\
        &  {\rm or}\,\, b\ne 0,\,\, b=c\,\, {\rm and}\,\, cd=1\\
        &    {\rm or}\,\, b\ne 0,\,\, b\ne c,\,\, cd=1\,\, {\rm and}\,\, db^2+cdb+c=0;\\
     2, &   {\rm if}\,\, b=0,\,\, cd\ne 1\,\, {\rm and}\,\, {\rm Tr}_{1}^n(\frac{1}{cd})=0 \\
        &  {\rm or}\,\, b\ne 0,\,\, b=c\,\,  {\rm and}\,\, {\rm Tr}_{1}^n(\frac{1}{cd})=0;\\
     0, &  {\rm otherwise}. \\
\end{cases}
$$
\end{itemize}
\end{lemma}

\begin{proof}
The proof methods for the LBCT and UBCT of $F(x)$ are similar, and we will omit this proof in the paper for brevity.
\end{proof}

Utilizing the results in Lemma \ref{inverse-UBCT} and Lemma \ref{inverse-LBCT}, we present a lemma we shall need later on.

\begin{lemma}\label{inverse-hard}
Let $F(x)=x^{2^n-2}$ be the inverse function over $\mathbb{F}_{2^n}$. Then $F(x)$ is hard when $n$ is odd.
\end{lemma}

\begin{proof}
According to the Lemma \ref{inverse-UBCT} and Lemma \ref{inverse-LBCT}, we have that the ${\rm dbct}(a,b,c,d)$ of $F(x)$ has values only if $b=c$ when $n$ is odd. By the Definition \ref{definition-Hard}, it is obvious that $F(x)$ is hard when $n$ is odd.
\end{proof}

With the above preparations, we give the DBCT of the inverse function in the following theorem.

\begin{theorem}\label{inverse-DBCT}
Let $F(x)=x^{2^n-2}$ be the inverse function over $\mathbb{F}_{2^n}$. For $a,d\in \gf_{2^n}$,
\begin{itemize}
\item [\rm (i)] if $a=0$ or $d=0$, then
$${\rm DBCT}_{F}(a,d)=2^{2n};$$
\item [\rm (i)] if $a\ne 0$ and $d\ne 0$, then
$${\rm DBCT}_{F}(a,d)=
\begin{cases}
     2^{n+1}, &  {\rm if}\,\, a=d;\\
     2^n+4,   &   {\rm if}\,\, a\ne d\,\, {\rm and}\,\, {\rm Tr}_{1}^n(\frac{a}{d})={\rm Tr}_{1}^n(\frac{d}{a})=0;\\
      2^n,  &   {\rm if}\,\, a\ne d,\,\, {\rm Tr}_{1}^n(\frac{a}{d})=0\,\, {\rm and}\,\, {\rm Tr}_{1}^n(\frac{d}{a})=1\\
            &    {\rm or}\,\, a\ne d,\,\, {\rm Tr}_{1}^n(\frac{a}{d})=1\,\, {\rm and}\,\, {\rm Tr}_{1}^n(\frac{d}{a})=0;\\
     2^n-4, &  {\rm if}\,\,a\ne d\,\, {\rm and}\,\, {\rm Tr}_{1}^n(\frac{a}{d})={\rm Tr}_{1}^n(\frac{d}{a})=1 \\
\end{cases}
$$
for $n>1$ is odd. For $n>2$ is even, if $n\equiv 0\,\, ({\rm mod}\,\,4)$, then
$${\rm DBCT}_{F}(a,d)=
\begin{cases}
     2^{n+1}+8, &  {\rm if}\,\, a=d;\\
     2^n+20, &  {\rm if}\,\,a\ne d\,\, {\rm and}\,\, a^3=d^3;\\
     2^n+4,   &   {\rm if}\,\,a^3\ne d^3\,\,{\rm and}\,\, {\rm Tr}_{1}^n(\frac{a}{d})={\rm Tr}_{1}^n(\frac{d}{a})=0;\\
      2^n,  &   {\rm if}\,\, a^3\ne d^3,\,\, {\rm Tr}_{1}^n(\frac{a}{d})=0\,\, {\rm and}\,\, {\rm Tr}_{1}^n(\frac{d}{a})=1\\
            &    {\rm or}\,\,a^3\ne d^3,\,\,  {\rm Tr}_{1}^n(\frac{a}{d})=1\,\, {\rm and}\,\, {\rm Tr}_{1}^n(\frac{d}{a})=0;\\
     2^n-4, &  {\rm if}\,\,a^3\ne d^3\,\, {\rm and}\,\, {\rm Tr}_{1}^n(\frac{a}{d})={\rm Tr}_{1}^n(\frac{d}{a})=1,\\
\end{cases}
$$
if $n\equiv 2\,\, ({\rm mod}\,\,4)$, then
$${\rm DBCT}_{F}(a,d)=
\begin{cases}
     2^{n+1}+8, &  {\rm if}\,\, a=d;\\
     2^n+12, &  {\rm if}\,\,a\ne d\,\, {\rm and}\,\, a^3=d^3;\\
     2^n-4,   &   {\rm if}\,\,a^3\ne d^3\,\, {\rm and}\,\, {\rm Tr}_{1}^n(\frac{a}{d})={\rm Tr}_{1}^n(\frac{d}{a})=1;\\
      2^n,  &   {\rm if}\,\, a^3\ne d^3,\,\, {\rm Tr}_{1}^n(\frac{a}{d})=0\,\, {\rm and}\,\, {\rm Tr}_{1}^n(\frac{d}{a})=1\\
            &    {\rm or}\,\,a^3\ne d^3,\,\,  {\rm Tr}_{1}^n(\frac{a}{d})=1\,\, {\rm and}\,\, {\rm Tr}_{1}^n(\frac{d}{a})=0;\\
     2^n+4, &  {\rm if}\,\,a^3\ne d^3\,\, {\rm and}\,\,{\rm Tr}_{1}^n(\frac{a}{d})={\rm Tr}_{1}^n(\frac{d}{a})=0.\\
\end{cases}
$$
\end{itemize}
\end{theorem}

\begin{proof}
According to the Property \ref{pro-DBCT}, for $a=0$ or $d=0$, we have
$$
{\rm DBCT}(0,d)=\sum\limits_{c}{\rm UBCT}_{F}(0,0,c)\cdot {\rm LBCT}_{F}(0,c,d)=2^{2n},
$$
$$
{\rm DBCT}(a,0)=\sum\limits_{b}{\rm UBCT}_{F}(a,b,0)\cdot {\rm LBCT}_{F}(b,0,0)=2^{2n}.
$$
Next, we assume that $a\ne 0$ and $d\ne 0$. We first consider the case of $n>1$ is odd.

By Lemma \ref{inverse-hard}, we have that $F(x)$ is hard when $n$ is odd. Then we have
\begin{equation*}
\begin{aligned}
{\rm DBCT}_{F}(a,d)&=\sum\limits_{b} {\rm DDT}_{F}(a,b)\cdot {\rm DDT}_{F}(b,d)\\
&=\sum\limits_{b\ne 0} {\rm DDT}_{F}(a,b)\cdot {\rm DDT}_{F}(b,d)+\sum\limits_{b=0} {\rm DDT}_{F}(a,b)\cdot {\rm DDT}_{F}(b,d)\\
&=\sum\limits_{b\ne 0} {\rm DDT}_{F}(a,b)\cdot {\rm DDT}_{F}(b,d),
\end{aligned}
\end{equation*}
since ${\rm DDT}_{F}(a,0)=0$ and ${\rm DDT}_{F}(0,b)=0$. Thus the entries of DBCT can be computed from DDT entries.
By the DDT of $F(x)$ in Lemma \ref{inverse-DDT}, 
for $a,d\in \gf_{2^n}^*$, we discuss the following two cases.

{\textbf{Case 1:}} If $a=d$, assume that there exists a $b\in \gf_{2^n}^*$ such that $ab=bd=1$. Then we have
\begin{equation*}
\begin{aligned}
{\rm DBCT}_{F}(a,a)&={\rm DDT}_{F}(a,\frac{1}{a})\cdot {\rm DDT}_{F}(\frac{1}{a},a)+\sum\limits_{b \in \gf_{2^n}\backslash \{0, \frac{1}{a}\}} {\rm DDT}_{F}(a,b)\cdot {\rm DDT}_{F}(b,a)\\
&=2\times 2+2\times 2\times |\{b \in \gf_{2^n}\backslash \{0, \frac{1}{a}\}: {\rm Tr}_1^n(\frac{1}{ab})=0\}|\\
&=4+4\times (2^{n-1}-1)\\
&=2^{n+1}.
\end{aligned}
\end{equation*}

{\textbf{Case 2:}} If $a\ne d$, assume that there exists a $b\in \gf_{2^n}^*$ such that $ab=1$ or $bd=1$. Then we have
\begin{equation}\label{d-n-odd}
\begin{aligned}
{\rm DBCT}_{F}(a,d)=&{\rm DDT}_{F}(a,\frac{1}{a})\cdot {\rm DDT}_{F}(\frac{1}{a},d)+{\rm DDT}_{F}(a,\frac{1}{d})\cdot {\rm DDT}_{F}(\frac{1}{d},d)\\
&+\sum\limits_{b \in \gf_{2^n}\backslash \{0, \frac{1}{a}, \frac{1}{d}\}} {\rm DDT}_{F}(a,b)\cdot {\rm DDT}_{F}(b,d)\\
=&{\rm DDT}_{F}(a,\frac{1}{a})\cdot {\rm DDT}_{F}(\frac{1}{a},d)+{\rm DDT}_{F}(a,\frac{1}{d})\cdot {\rm DDT}_{F}(\frac{1}{d},d)\\
&+2\times 2\times |\{b \in \gf_{2^n}\backslash \{0, \frac{1}{a}, \frac{1}{d}\}:{\rm Tr}_1^n(\frac{1}{ab})=0\,\, {\rm and}\,\,{\rm Tr}_1^n(\frac{1}{bd})=0\}|.\\
\end{aligned}
\end{equation}

We define
$$
N(b)=\{b \in \gf_{2^n}\backslash \{0, \frac{1}{a}, \frac{1}{d}\}: {\rm Tr}_{1}^n(\frac{1}{ab})=0\,\, {\rm and}\,\, {\rm Tr}_{1}^n(\frac{1}{bd})=0\},
$$
then we have
\begin{equation*}
\begin{aligned}
|N(b)|=&\sum\limits_{b\in \gf_{2^n}}\Big(\frac{1}{2}\sum\limits_{y_1\in \gf_{2}}((-1)^{{\rm Tr}_{1}^n(\frac{1}{bd})\cdot y_1})\cdot \frac{1}{2}\sum\limits_{y_2\in \gf_{2}}((-1)^{{\rm Tr}_{1}^n(\frac{1}{ab})\cdot y_2})\Big)\\
&-(\frac{1}{2}\sum\limits_{y_1\in \gf_{2}}((-1)^{{\rm Tr}_{1}^n(1)\cdot y_1})(\frac{1}{2}\sum\limits_{y_2\in \gf_{2}}((-1)^{{\rm Tr}_{1}^n(\frac{a}{d})\cdot y_2})\\
&-(\frac{1}{2}\sum\limits_{y_1\in \gf_{2}}((-1)^{{\rm Tr}_{1}^n(\frac{d}{a})\cdot y_1})(\frac{1}{2}\sum\limits_{y_2\in \gf_{2}}((-1)^{{\rm Tr}_{1}^n(1)\cdot y_2})-1\\
=&\frac{1}{4}\sum\limits_{b\in \gf_{2^n}}\Big((1+(-1)^{{\rm Tr}_{1}^n(\frac{1}{ab})})(1+(-1)^{{\rm Tr}_{1}^n(\frac{1}{db})})\Big)-\frac{1}{2}(1-1)-\frac{1}{2}(1-1)-1\\
=&2^{n-2}+\frac{1}{4}(\sum\limits_{b\in \gf_{2^n}}(-1)^{{\rm Tr}_{1}^n(\frac{1}{bd})}+\sum\limits_{b\in \gf_{2^n}}(-1)^{{\rm Tr}_{1}^n(\frac{1}{ab})}+\sum\limits_{b\in \gf_{2^n}}(-1)^{{\rm Tr}_{1}^n(\frac{1}{ab}+\frac{1}{bd})})-1\\
=&2^{n-2}-1.
\end{aligned}
\end{equation*}

In the sequel we should distinguish three subcases.

{\textbf{Subcase 1:}} If ${\rm Tr}_{1}^n(\frac{a}{d})={\rm Tr}_{1}^n(\frac{d}{a})=0$. By (\ref{d-n-odd}), we have ${\rm DDT}_{F}(\frac{1}{a},d)={\rm DDT}_{F}(a,\frac{1}{d})=2$. Thus we obtain
\begin{equation*}
\begin{aligned}
{\rm DBCT}_{F}(a,d)&=2\times 2\times 2+(2^{n-2}-1)\times 2\times 2\\
&=2^n+4.
\end{aligned}
\end{equation*}

{\textbf{Subcase 2:}} If ${\rm Tr}_{1}^n(\frac{a}{d})=0$ and ${\rm Tr}_{1}^n(\frac{d}{a})=1$ or ${\rm Tr}_{1}^n(\frac{a}{d})=1$ and ${\rm Tr}_{1}^n(\frac{d}{a})=0$. By (\ref{d-n-odd}), we have ${\rm DDT}_{F}(\frac{1}{a},d)=2$ and ${\rm DDT}_{F}(a,\frac{1}{d})=0$ or ${\rm DDT}_{F}(\frac{1}{a},d)=0$ and ${\rm DDT}_{F}(a,\frac{1}{d})=2$. Then we further have
\begin{equation*}
\begin{aligned}
{\rm DBCT}_{F}(a,d)&=2\times 2+(2^{n-2}-1)\times 2\times 2\\
&=2^n.
\end{aligned}
\end{equation*}

{\textbf{Subcase 3:}} If ${\rm Tr}_{1}^n(\frac{a}{d})={\rm Tr}_{1}^n(\frac{d}{a})=1$. By (\ref{d-n-odd}), we have ${\rm DDT}_{F}(\frac{1}{a},d)={\rm DDT}_{F}(a,\frac{1}{d})=0$. Hence, we get
\begin{equation*}
\begin{aligned}
{\rm DBCT}_{F}(a,d)&=(2^{n-2}-1)\times 2\times 2\\
&=2^n-4.
\end{aligned}
\end{equation*}

Now we consider the case of $n>2$ is even.

By Lemmas \ref{inverse-UBCT} and \ref{inverse-LBCT}, we have that $F(x)$ is not hard when $n$ is even. According to Definition \ref{definition-DBCT}, we have
$$
\begin{aligned}
{\rm DBCT}_{F}(a,d)&=\sum\limits_{b=c} {\rm dbct}(a,b,c,d)+\sum\limits_{b\ne c} {\rm dbct}(a,b,c,d)\\
&=\sum\limits_{b=c} {\rm UBCT}(a,b,c)\cdot {\rm LBCT}(b,c,d)+\sum\limits_{b\ne c} {\rm UBCT}(a,b,c)\cdot {\rm LBCT}(b,c,d).
\end{aligned}
$$
Thus the DBCT of $F(x)$ can be expressed as the sum of two parts.

Firstly, we consider the sum of the part when $b\ne c$.

For $a,d \in \gf_{2^n}^*$ with $b\ne c$, Lemmas \ref{inverse-UBCT} and \ref{inverse-LBCT} establish that ${\rm UBCT}(a,b,c)=4$ only if $ab=1$, $ac^2+abc+b=0$; otherwise, ${\rm UBCT}(a,b,c)=0$. Similarly, ${\rm LBCT}(b,c,d)=4$ iff $cd=1$, $db^2+dcb+c=0$; otherwise ${\rm LBCT}(b,c,d)=0$. Therefore, both UBCT and LBCT of $F(x)$ take non-zero values simultaneously when $ab=1$, $cd=1$, $ac^2+abc+b=0$ and $db^2+dcb+c=0$.
Under these conditions, it is evident that $a\ne d$ and $a^3=d^3$.
Let $\omega$ be a primitive element of $\gf_{2^n}$ and assume that $a=\omega^i$, $d=\omega^j$ with $i,j\in [0\dots 2^n-2]$ and $i\ne j$. The equation $a^3=d^3$ is equivalent to the congruence $3j\equiv 3i\,\, ({\rm mod}\,\, 2^n-1)$, which can be reduced to
$3(j-i)\equiv 0\,\, ({\rm mod}\,\, 2^n-1)$.
Since ${\rm gcd}(3, 2^n-1)=3$ for even $n$, besides $i=j$, the above equation also has two other solutions in $\gf_{2^n}$: $j-i=\frac{2^n-1}{3}$ and $j-i=\frac{2(2^n-1)}{3}$. Hence, for a given $d$, there are two distinct values of $a$ such that $a^3=d^3$.
For $a,d\in \gf_{2^n}^*$ with $a\ne d$, let $s=\frac{2^n-1}{3}$. Then we have $\omega^{-s}=\omega^{3\cdot \frac{2^n-1}{3}}\cdot \omega^{-\frac{2^n-1}{3}}=\omega^{2\cdot \frac{2^n-1}{3}}=\omega^{2s}$. Similarly, we also get $\omega^{-2s}=\omega^s$.
Hence, we observe that ${\rm Tr}_{1}^n(\omega^{-2s})={\rm Tr}_{1}^n(\omega^s)={\rm Tr}_{1}^n(\omega^{2s})={\rm Tr}_{1}^n(\omega^{-s})={\rm Tr}_{1}^n(\frac{a}{d})={\rm Tr}_{1}^n(\frac{d}{a})$. This implies that ${\rm Tr}_{1}^n(\frac{d}{a})$ and ${\rm Tr}_{1}^n(\frac{a}{d})$ are either both equal to $0$ or $1$ when $a^3=d^3$ and $a\ne d$. Let $\omega^s=\alpha$, then we obtain that $\alpha^3=1$. This implies that $\alpha^2+\alpha+1=0$ since $\alpha\ne 1$. Next we consider the value of ${\rm Tr}_{1}^n(\alpha)$. Since
$$
\begin{aligned}
{\rm Tr}_{1}^n(\alpha)&=\alpha+\alpha^2+\alpha^{2^2}+\alpha^{2^3}+\cdots+\alpha^{2^{n-1}}\\
&=(\alpha+\alpha^2)+(\alpha+\alpha^2)^{2^2}+(\alpha+\alpha^2)^{2^4}+\cdots+(\alpha+\alpha^2)^{2^{2i}},
\end{aligned}
$$
where $i$ is a positive integer. Then we have $2i+1=n-1$, that is, $n=2i+2$. One immediately concludes that ${\rm Tr}_{1}^n(\alpha)=1$ and ${\rm Tr}_{1}^n(\alpha)=0$ when $i$ is even and $i$ is odd, respectively. Summarizing the above discussion, when $a^3=d^3$ and $a\ne d$,
we have ${\rm Tr}_{1}^n(\frac{a}{d})={\rm Tr}_{1}^n(\frac{d}{a})=1$ when $n\equiv 2\,\, ({\rm mod}\,\, 4)$ and ${\rm Tr}_{1}^n(\frac{a}{d})={\rm Tr}_{1}^n(\frac{d}{a})=0$ when $n\equiv 0\,\, ({\rm mod}\,\, 4)$.
If $a$ and $d$ satisfy the above conditions, then we have
$$
{\rm DBCT}_{F}(a,d)=\sum\limits_{b\ne c} {\rm UBCT}(a,b,c)\cdot {\rm LBCT}(b,c,d)=4\times 4=16.
$$

Next, we consider the sum of the part when $b=c$.

For $a,d \in \gf_{2^n}^*$, if $b=c$, we have
$$
\begin{aligned}
{\rm DBCT}_{F}(a,d)&=\sum\limits_{b=c} {\rm UBCT}(a,b,c)\cdot {\rm LBCT}(b,c,d)\\
&=\sum\limits_{b} {\rm DDT}_{F}(a,b)\cdot {\rm DDT}_{F}(b,d)\\
\end{aligned}
$$
This means that in this case the UBCT and LBCT for computing DBCT degenerate to DDT.

Since ${\rm DDT}_{F}(a,0)=0$ and ${\rm DDT}_{F}(0,b)=0$, we assume that $a,b,d\ne 0$ as follows. By the DDT of $F(x)$ in Lemma \ref{inverse-DDT},
for $a,d\in \gf_{2^n}^*$, we discuss the following two cases.

{\textbf{Case 1:}} If $a=d$, assume that there exists a $b\in \gf_{2^n}^*$ such that $ab=bd=1$. Then we have
\begin{equation*}
\begin{aligned}
{\rm DBCT}_{F}(a,a)&={\rm DDT}_{F}(a,\frac{1}{a})\cdot {\rm DDT}_{F}(\frac{1}{a},a)+\sum\limits_{b \in \gf_{2^n}\backslash \{0, \frac{1}{a}\}} {\rm DDT}_{F}(a,b)\cdot {\rm DDT}_{F}(b,a)\\
&=4\times 4+2\times 2\times |\{b \in \gf_{2^n}\backslash \{0, \frac{1}{a}\}: {\rm Tr}_1^n(\frac{1}{ab})=0\}|\\
&=16+4\times (2^{n-1}-2)\\
&=2^{n+1}+8.
\end{aligned}
\end{equation*}

{\textbf{Case 2:}} If $a\ne d$, assume that there exists a $b\in \gf_{2^n}^*$ such that $ab=1$ or $bd=1$. Then we get
\begin{equation}\label{d-n-even}
\begin{aligned}
{\rm DBCT}_{F}(a,d)=&{\rm DDT}_{F}(a,\frac{1}{a})\cdot {\rm DDT}_{F}(\frac{1}{a},d)+{\rm DDT}_{F}(a,\frac{1}{d})\cdot {\rm DDT}_{F}(\frac{1}{d},d)\\
&+\sum\limits_{b \in \gf_{2^n}\backslash \{0, \frac{1}{a}, \frac{1}{d}\}} {\rm DDT}_{F}(a,b)\cdot {\rm DDT}_{F}(b,a)\\
=&{\rm DDT}_{F}(a,\frac{1}{a})\cdot {\rm DDT}_{F}(\frac{1}{a},d)+{\rm DDT}_{F}(a,\frac{1}{d})\cdot {\rm DDT}_{F}(\frac{1}{d},d)\\
&+2\times 2\times |\{b \in \gf_{2^n}\backslash \{0, \frac{1}{a}, \frac{1}{d}\}:{\rm Tr}_1^n(\frac{1}{ab})=0\,\, {\rm and}\,\,{\rm Tr}_1^n(\frac{1}{bd})=0\}|.\\
\end{aligned}
\end{equation}

We also define
$$
N(b)=\{b \in \gf_{2^n}\backslash \{0, \frac{1}{a}, \frac{1}{d}\}: {\rm Tr}_{1}^n(\frac{1}{ab})=0\,\, {\rm and}\,\, {\rm Tr}_{1}^n(\frac{1}{bd})=0\},
$$
then we obtain
\begin{equation*}
\begin{aligned}
|N(b)|=&\sum\limits_{b\in \gf_{2^n}}\Big(\frac{1}{2}\sum\limits_{y_1\in \gf_{2}}((-1)^{{\rm Tr}_{1}^n(\frac{1}{bd})\cdot y_1})\cdot \frac{1}{2}\sum\limits_{y_2\in \gf_{2}}((-1)^{{\rm Tr}_{1}^n(\frac{1}{ab})\cdot y_2})\Big)\\
&-(\frac{1}{2}\sum\limits_{y_1\in \gf_{2}}((-1)^{{\rm Tr}_{1}^n(1)\cdot y_1})(\frac{1}{2}\sum\limits_{y_2\in \gf_{2}}((-1)^{{\rm Tr}_{1}^n(\frac{a}{d})\cdot y_2})\\
&-(\frac{1}{2}\sum\limits_{y_1\in \gf_{2}}((-1)^{{\rm Tr}_{1}^n(\frac{d}{a})\cdot y_1})(\frac{1}{2}\sum\limits_{y_2\in \gf_{2}}((-1)^{{\rm Tr}_{1}^n(1)\cdot y_2})-1\\
=&\frac{1}{4}\sum\limits_{b\in \gf_{2^n}}\Big((1+(-1)^{{\rm Tr}_{1}^n(\frac{1}{ab})})(1+(-1)^{{\rm Tr}_{1}^n(\frac{1}{db})})\Big)\\
&-\frac{1}{2}(1+(-1)^{{\rm Tr}_{1}^n(\frac{d}{a})})-\frac{1}{2}(1+(-1)^{{\rm Tr}_{1}^n(\frac{a}{d})})-1\\
=&2^{n-2}+\frac{1}{4}(\sum\limits_{b\in \gf_{2^n}}(-1)^{{\rm Tr}_{1}^n(\frac{1}{bd})}+\sum\limits_{b\in \gf_{2^n}}(-1)^{{\rm Tr}_{1}^n(\frac{1}{ab})}+\sum\limits_{b\in \gf_{2^n}}(-1)^{{\rm Tr}_{1}^n(\frac{1}{ab}+\frac{1}{bd})})\\
&-\frac{1}{2}((-1)^{{\rm Tr}_{1}^n(\frac{a}{d})}+(-1)^{{\rm Tr}_{1}^n(\frac{d}{a})})-2\\
=&2^{n-2}-2-\frac{1}{2}((-1)^{{\rm Tr}_{1}^n(\frac{a}{d})}+(-1)^{{\rm Tr}_{1}^n(\frac{d}{a})}).
\end{aligned}
\end{equation*}
The values of $|N(b)|$ can be considered the following three cases.

\begin{itemize}
\item [\rm (i)] If ${\rm Tr}_{1}^n(\frac{a}{d})={\rm Tr}_{1}^n(\frac{d}{a})=0$, then we have
$$
|N(b)|=2^{n-2}-2-1=2^{n-2}-3.
$$

\item [\rm (ii)] If ${\rm Tr}_{1}^n(\frac{a}{d})=0$ and ${\rm Tr}_{1}^n(\frac{d}{a})=1$ or ${\rm Tr}_{1}^n(\frac{a}{d})=1$ and ${\rm Tr}_{1}^n(\frac{d}{a})=0$, then we have
$$
|N(b)|=2^{n-2}-2.
$$

\item [\rm (iii)] If ${\rm Tr}_{1}^n(\frac{a}{d})={\rm Tr}_{1}^n(\frac{d}{a})=1$, then we have
$$
|N(b)|=2^{n-2}-2+1=2^{n-2}-1.
$$
\end{itemize}

Summarizing the above discussion, we can conclude the following results:

{\textbf{Subcase 1:}} If $n\equiv 0\,\, ({\rm mod}\,\, 4)$. According to the discussion on the sum of the part when $b\ne c$, it follows that ${\rm Tr}_{1}^n(\frac{a}{d})={\rm Tr}_{1}^n(\frac{d}{a})=0$ whenever $a^3=d^3$ and $a\ne d$. This implies that ${\rm dbct}(a,b,c,d)$ take values for both $b\ne c$ and $b=c$ under the condition $a^3=d^3$ and $a\ne d$, otherwise, it has values only when $b=c$. Next, our discussion proceeds with three cases.

{\textbf{Subcase 1.1:}} If $a^3=d^3$ and $a\ne d$, we always have ${\rm Tr}_{1}^n(\frac{a}{d})={\rm Tr}_{1}^n(\frac{d}{a})=0$. By (\ref{d-n-even}), we have ${\rm DDT}_{F}(\frac{1}{a},d)={\rm DDT}_{F}(a,\frac{1}{d})=2$.
Then we get
$$
{\rm DBCT}_{F}(a,d)=4\times 2\times 2+(2^{n-2}-3)\times 2\times 2+16=2^n+20,
$$
if $a^3\ne d^3$ and ${\rm Tr}_{1}^n(\frac{a}{d})={\rm Tr}_{1}^n(\frac{d}{a})=0$, it can be easily obtained that
$$
{\rm DBCT}_{F}(a,d)=4\times 2\times 2+(2^{n-2}-3)\times 2\times 2=2^n+4,
$$

{\textbf{Subcase 1.2:}} If ${\rm Tr}_{1}^n(\frac{a}{d})=0$ and ${\rm Tr}_{1}^n(\frac{d}{a})=1$ or ${\rm Tr}_{1}^n(\frac{a}{d})=1$ and ${\rm Tr}_{1}^n(\frac{d}{a})=0$. By (\ref{d-n-even}), ${\rm DDT}_{F}(\frac{1}{a},d)$ and ${\rm DDT}_{F}(a,\frac{1}{d})$ either take the values of $2$ and $0$, respectively, or $0$ and $2$, respectively. Then we further have
$$
{\rm DBCT}_{F}(a,d)=4\times 2+(2^{n-2}-2)\times 2\times 2=2^n,
$$

{\textbf{Subcase 1.3:}} If ${\rm Tr}_{1}^n(\frac{a}{d})={\rm Tr}_{1}^n(\frac{d}{a})=1$. By (\ref{d-n-even}), we have ${\rm DDT}_{F}(\frac{1}{a},d)={\rm DDT}_{F}(a,\frac{1}{d})=0$.
Hence, we get
$$
{\rm DBCT}_{F}(a,d)=(2^{n-2}-1)\times 2\times 2=2^n-4.
$$

{\textbf{Subcase 2:}} If $n\equiv 2\,\, ({\rm mod}\,\, 4)$. According to the discussion on the sum of the part when $b\ne c$, we always have ${\rm Tr}_{1}^n(\frac{a}{d})={\rm Tr}_{1}^n(\frac{d}{a})=1$ when $a^3=d^3$ and $a\ne d$. Hence, this means that ${\rm dbct}(a,b,c,d)$ has values for both $b\ne c$ and $b=c$ when $a^3=d^3$ and $a\ne d$, and otherwise ${\rm dbct}(a,b,c,d)$ has values only for $b=c$.
Next, we discuss the following three cases.

{\textbf{Subcase 2.1:}} If $a^3=d^3$ and $a\ne d$. By (\ref{d-n-even}), we have ${\rm DDT}_{F}(\frac{1}{a},d)={\rm DDT}_{F}(a,\frac{1}{d})=0$.
Then we further have
$$
{\rm DBCT}_{F}(a,d)=(2^{n-2}-1)\times 2\times 2+16=2^n+12,
$$
and if and $a^3\ne d^3$ and ${\rm Tr}_{1}^n(\frac{a}{d})={\rm Tr}_{1}^n(\frac{d}{a})=1$. Then we further get
$$
{\rm DBCT}_{F}(a,d)=(2^{n-2}-1)\times 2\times 2=2^n-4,
$$

{\textbf{Subcase 2.2:}} Similarly to the Subcase 1.2 for $n\equiv 0\,\, ({\rm mod}\,\, 4)$, here we omit the discussion.

{\textbf{Subcase 2.3:}} If ${\rm Tr}_{1}^n(\frac{a}{d})={\rm Tr}_{1}^n(\frac{d}{a})=0$. By (\ref{d-n-even}), we have ${\rm DDT}_{F}(\frac{1}{a},d)={\rm DDT}_{F}(a,\frac{1}{d})=2$.
Then we further obtain
$$
{\rm DBCT}_{F}(a,d)=4\times 2\times 2+(2^{n-2}-3)\times 2\times 2=2^n+4.
$$
This completes the proof.
\end{proof}

According to the Theorem \ref{inverse-DBCT}, we can obtain the following result.

\begin{proposition}\label{proposotion-inverse-DBCT}
Let $F(x)=x^{2^n-2}$ be the inverse function over $\mathbb{F}_{2^n}$. For $a, d\in \gf_{2^n}$, the explicit values of all entries of the {\rm DBCT} of $F(x)$ satisfies
\begin{equation*}
\begin{aligned}
&|\{(a,b)\in \gf_{2^n}^2: {\rm DBCT}_{F}(a,d)=2^{2n}\}|=2^{n+1}-1;\\
&|\{(a,b)\in \gf_{2^n}^2: {\rm DBCT}_{F}(a,d)=2^{n+1}\}|=2^n-1;\\
&|\{(a,b)\in \gf_{2^n}^2: {\rm DBCT}_{F}(a,d)=2^n+4\}|=(2^{n-2}+\frac{K(1)}{4}-1)(2^n-1);\\
&|\{(a,b)\in \gf_{2^n}^2: {\rm DBCT}_{F}(a,d)=2^n\}|=(2^{n-1}-\frac{K(1)}{2})(2^n-1);\\
&|\{(a,b)\in \gf_{2^n}^2: {\rm DBCT}_{F}(a,d)=2^n-4\}|=(2^{n-2}+\frac{K(1)}{4}-1)(2^n-1),\\
\end{aligned}
\end{equation*}
for $n>1$ is odd, and for $n>2$ is even, if $n\equiv 0\,\, ({\rm mod}\,\,4)$, then
\begin{equation*}
\begin{aligned}
&|\{(a,b)\in \gf_{2^n}^2: {\rm DBCT}_{F}(a,d)=2^{2n}\}|=2^{n+1}-1;\\
&|\{(a,b)\in \gf_{2^n}^2: {\rm DBCT}_{F}(a,d)=2^{n+1}+8\}|=2^n-1;\\
&|\{(a,b)\in \gf_{2^n}^2: {\rm DBCT}_{F}(a,d)=2^n+20\}|=2(2^n-1);\\
&|\{(a,b)\in \gf_{2^n}^2: {\rm DBCT}_{F}(a,d)=2^n+4\}|=(2^{n-2}+\frac{K(1)}{4}-4)(2^n-1);\\
&|\{(a,b)\in \gf_{2^n}^2: {\rm DBCT}_{F}(a,d)=2^n\}|=(2^{n-1}-\frac{K(1)}{2})(2^n-1);\\
&|\{(a,b)\in \gf_{2^n}^2: {\rm DBCT}_{F}(a,d)=2^n-4\}|=(2^{n-2}+\frac{K(1)}{4})(2^n-1),\\
\end{aligned}
\end{equation*}
if $n\equiv 2\,\, ({\rm mod}\,\,4)$, then
\begin{equation*}
\begin{aligned}
&|\{(a,b)\in \gf_{2^n}^2: {\rm DBCT}_{F}(a,d)=2^{2n}\}|=2^{n+1}-1;\\
&|\{(a,b)\in \gf_{2^n}^2: {\rm DBCT}_{F}(a,d)=2^{n+1}+8\}|=2^n-1;\\
&|\{(a,b)\in \gf_{2^n}^2: {\rm DBCT}_{F}(a,d)=2^n+12\}|=2(2^n-1);\\
&|\{(a,b)\in \gf_{2^n}^2: {\rm DBCT}_{F}(a,d)=2^n-4\}|=(2^{n-2}+\frac{K(1)}{4}-2)(2^n-1);\\
&|\{(a,b)\in \gf_{2^n}^2: {\rm DBCT}_{F}(a,d)=2^n\}|=(2^{n-1}-\frac{K(1)}{2})(2^n-1);\\
&|\{(a,b)\in \gf_{2^n}^2: {\rm DBCT}_{F}(a,d)=2^n+4\}|=(2^{n-2}+\frac{K(1)}{4}-2)(2^n-1).\\
\end{aligned}
\end{equation*}
\end{proposition}

\begin{proof}
By Theorem \ref{inverse-DBCT}, ${\rm DBCT}_{F}(a,d)=2^{2n}$ happens for $a=0$ or $d=0$. Hence,
$$
|\{(a,b)\in \gf_{2^n}^2: {\rm DBCT}_{F}(a,d)=2^{2n}\}|=2^n+2^n-1=2^{n+1}-1.
$$

Next, we divide the discussion into two cases as follows.

{\textbf{Case 1:}} If $n>1$ is odd. By Theorem \ref{inverse-DBCT}, we have ${\rm DBCT}_{F}(a,d)=2^{n+1}$ only if $a=d$. This gives
$$
|\{(a,b)\in \gf_{2^n}^2: {\rm DBCT}_{F}(a,d)=2^{n+1}\}|=2^n-1.
$$
For $a,d\in \gf_{2^n}^*$, in what follows, we always assume that $a\ne d$ unless otherwise stated. We have already shown that ${\rm DBCT}_{F}(a,d)=2^n+4$ if ${\rm Tr}_{1}^n(\frac{a}{d})={\rm Tr}_{1}^n(\frac{d}{a})=0$. Let $\frac{a}{d}=x$, then we have $x\in \gf_{2^n}\backslash \gf_{2}$.
By the definition of the Kloosterman sum $K(1)$ in Definition \ref{definition-K}, we have
\begin{equation*}
\begin{aligned}
K(1)-2&=\sum\limits_{x\in \gf_{2^n}\backslash \gf_{2}}(-1)^{{\rm Tr}_{1}^n(\frac{1}{x}+x)}\\
&=2|\{x\in \gf_{2^n}\backslash \gf_{2}: {\rm Tr}_{1}^n(\frac{1}{x}+x)=0\}|-(2^n-2)\\
&=4|\{x\in \gf_{2^n}\backslash \gf_{2}: {\rm Tr}_{1}^n(\frac{1}{x})=0\,\, {\rm and}\,\, {\rm Tr}_{1}^n(x)=0\}|-(2^n-2).
\end{aligned}
\end{equation*}
Thus,
$$
|\{x\in \gf_{2^n}\backslash \gf_{2}: {\rm Tr}_{1}^n(\frac{1}{x})=0\,\, {\rm and}\,\, {\rm Tr}_{1}^n(x)=0\}|=2^{n-2}+\frac{K(1)}{4}-1.
$$
Similarly, we also get
$$
|\{x\in \gf_{2^n}\backslash \gf_{2}: {\rm Tr}_{1}^n(\frac{1}{x})=1\,\, {\rm and}\,\, {\rm Tr}_{1}^n(x)=1\}|=2^{n-2}+\frac{K(1)}{4}-1.
$$
Then it can be easily obtained that
$$
\begin{aligned}
&|\{x\in \gf_{2^n}\backslash \gf_{2}: {\rm Tr}_{1}^n(\frac{1}{x})=0\,\, {\rm and}\,\, {\rm Tr}_{1}^n(x)=1\,\, {\rm or}\,\, {\rm Tr}_{1}^n(\frac{1}{x})=1\,\, {\rm and}\,\, {\rm Tr}_{1}^n(x)=0\}|\\
&=2^n-2-(2^{n-2}+\frac{K(1)}{4}-1)-(2^{n-2}+\frac{K(1)}{4}-1)\\
&=2^{n-1}-\frac{K(1)}{2}.
\end{aligned}
$$
Hence, we have the desired result in Proposition \ref{proposotion-inverse-DBCT} when $n>1$ is odd.

{\textbf{Case 2:}} If $n>2$ is even. By Theorem \ref{inverse-DBCT}, we have ${\rm DBCT}_{F}(a,d)=2^{n+1}+8$ only if $a=d$. Then we get
$$
|\{(a,b)\in \gf_{2^n}^2: {\rm DBCT}_{F}(a,d)=2^{n+1}\}|=2^n-1.
$$
For $a,d\in \gf_{2^n}^*$, in what follows, we always assume that $a\ne d$ unless otherwise stated. When $n\equiv 0\,\, ({\rm mod}\,\, 4)$, we have shown that ${\rm DBCT}_{F}(a,d)=2^n+20$ if $a^3=d^3$. According to the discussion on the sum of the part when $b\ne c$, note that the number of $a$ such that $a^3=d^3$ is two for a given element $d$.
Thus we get
$$
|\{(a,b)\in \gf_{2^n}^2: {\rm DBCT}_{F}(a,b)=2^n+20\}|=2(2^n-1).
$$
Also, when $a^3\ne d^3$, we have shown that ${\rm DBCT}_{F}(a,d)=2^n+4$ and ${\rm DBCT}_{F}(a,d)=2^n-4$ when ${\rm Tr}_{1}^n(\frac{a}{d})={\rm Tr}_{1}^n(\frac{d}{a})=0$ and ${\rm Tr}_{1}^n(\frac{a}{d})={\rm Tr}_{1}^n(\frac{d}{a})=1$, respectively.
Let $\frac{a}{d}=x$ and $s=\frac{2^n-1}{3}$. According to the discussion on the sum of the part when $b\ne c$, we have $x\in \gf_{2^n}\backslash \{0,1,\omega^s, \omega^{2s}\}$. Next, we consider the number of $x\in \gf_{2^n}\backslash \{0,1,\omega^s, \omega^{2s}\}$ such that ${\rm Tr}_{1}^n(\frac{1}{x})={\rm Tr}_{1}^n(x)=0$ and ${\rm Tr}_{1}^n(\frac{1}{x})={\rm Tr}_{1}^n(x)=1$, respectively. We define
$$
A_1:=\{x\in \gf_{2^n}\backslash \{0,1,\omega^s, \omega^{2s}\}: {\rm Tr}_{1}^n(\frac{1}{x})=0\,\, {\rm and}\,\, {\rm Tr}_{1}^n(x)=0\};
$$
$$
A_2:=\{x\in \gf_{2^n}\backslash \{0,1,\omega^s, \omega^{2s}\}: {\rm Tr}_{1}^n(\frac{1}{x})=0\,\, {\rm and}\,\, {\rm Tr}_{1}^n(x)=1\};
$$
$$
A_3:=\{x\in \gf_{2^n}\backslash \{0,1,\omega^s, \omega^{2s}\}: {\rm Tr}_{1}^n(\frac{1}{x})=1\,\, {\rm and}\,\, {\rm Tr}_{1}^n(x)=0\};
$$
$$
A_4:=\{x\in \gf_{2^n}\backslash \{0,1,\omega^s, \omega^{2s}\}: {\rm Tr}_{1}^n(\frac{1}{x})=1\,\, {\rm and}\,\, {\rm Tr}_{1}^n(x)=1\}.
$$
When $n\equiv 0\,\, ({\rm mod}\,\, 4)$, we have known that ${\rm Tr}_{1}^n(\omega^s)={\rm Tr}_{1}^n(\omega^{2s})=0$. Thus we get $|A_1|+|A_2|=2^{n-1}-4$ and $|A_2|+|A_4|=2^{n-1}$, this implies that $|A_4|=|A_1|+4$. When $n\equiv 2\,\, ({\rm mod}\,\, 4)$, we have known that ${\rm Tr}_{1}^n(\omega^s)={\rm Tr}_{1}^n(\omega^{2s})=1$. Thus we get $|A_1|+|A_2|=2^{n-1}-2$ and $|A_2|+|A_4|=2^{n-1}-2$, this implies that $|A_4|=|A_1|$.
By the definition of the Kloosterman sum $K(1)$ in Definition \ref{definition-K}, we have
\begin{equation*}
\begin{split}
K(1)&=\sum_{x\in \gf_{2^n}}(-1)^{{\rm Tr}_{1}^n(\frac{1}{x}+x)}\\
&=\sum_{\substack{x\in \gf_{2^n} \backslash \{0,1,\omega^s, \omega^{2s}\}}}(-1)^{{\rm Tr}_{1}^n(\frac{1}{x}+x)}+1+1+(-1)^{{\rm Tr}_{1}^n(\omega^{-s}+\omega^s)}+(-1)^{{\rm Tr}_{1}^n(\omega^{-2s}+\omega^{2s})}\\
&=\sum_{\substack{x\in \gf_{2^n} \backslash \{0,1,\omega^s, \omega^{2s}\}}}(-1)^{{\rm Tr}_{1}^n(\frac{1}{x}+x)}+4\\
&=2|\{x\in \gf_{2^n} \backslash \{0,1,\omega^s, \omega^{2s}\}: {\rm Tr}_{1}^n(\frac{1}{x}+x)=0\}|-(2^n-4)+4\\
&=2(|A_1|+|A_4|)-2^n+8,
\end{split}
\end{equation*}
where the third equality holds since ${\rm Tr}_{1}^n(\omega^{-s})={\rm Tr}_{1}^n(\omega^{2s})={\rm Tr}_{1}^n(\omega^s)={\rm Tr}_{1}^n(\omega^{-2s})$.
Thus, when $n\equiv 0\,\, ({\rm mod}\,\, 4)$, we obtain
$$
|A_1|=2^{n-2}+\frac{K(1)}{4}-4,\,\, |A_4|=2^{n-2}+\frac{K(1)}{4}.
$$
When $n\equiv 2\,\, ({\rm mod}\,\, 4)$, we have
$$
|A_1|=2^{n-2}+\frac{K(1)}{4}-2,\,\, |A_4|=2^{n-2}+\frac{K(1)}{4}-2.
$$
According to the above discussion, we always have $|A_2|+|A_3|=2^{n-1}-\frac{K(1)}{2}$ whether $n\equiv 0\,\, ({\rm mod}\,\, 4)$ or $n\equiv 2\,\, ({\rm mod}\,\, 4)$.

Summarizing the above discussion, we can conclude the result in this proposition.
\end{proof}

Next we provide a numerical example.

\begin{example}
If $n=4$, then we have $K(1)=0$. By Magma, the {\rm DBCT} of $F(x)$ over $\mathbb{F}_{2^4}$ can be directly computed:
\begin{equation*}
\begin{aligned}
&|\{(a,b)\in \gf_{2^n}^2: {\rm DBCT}_{F}(a,d)=256\}|=31;\\
&|\{(a,b)\in \gf_{2^n}^2: {\rm DBCT}_{F}(a,d)=40\}|=15;\\
&|\{(a,b)\in \gf_{2^n}^2: {\rm DBCT}_{F}(a,d)=36\}|=30;\\
&|\{(a,b)\in \gf_{2^n}^2: {\rm DBCT}_{F}(a,d)=16\}|=120;\\
&|\{(a,b)\in \gf_{2^n}^2: {\rm DBCT}_{F}(a,d)=12\}|=60.\\
\end{aligned}
\end{equation*}
If $n=5$, then we have $K(1)=12$. By Magma, the {\rm DBCT} of $F(x)$ over $\mathbb{F}_{2^5}$ can be directly computed:
\begin{equation*}
\begin{aligned}
&|\{(a,b)\in \gf_{2^n}^2: {\rm DBCT}_{F}(a,d)=1024\}|=63;\\
&|\{(a,b)\in \gf_{2^n}^2: {\rm DBCT}_{F}(a,d)=64\}|=31;\\
&|\{(a,b)\in \gf_{2^n}^2: {\rm DBCT}_{F}(a,d)=36\}|=310;\\
&|\{(a,b)\in \gf_{2^n}^2: {\rm DBCT}_{F}(a,d)=32\}|=310;\\
&|\{(a,b)\in \gf_{2^n}^2: {\rm DBCT}_{F}(a,d)=28\}|=310.\\
\end{aligned}
\end{equation*}
If $n=6$, then we have $K(1)=-8$. By Magma, the {\rm DBCT} of $F(x)$ over $\mathbb{F}_{2^6}$ can be directly computed:
\begin{equation*}
\begin{aligned}
&|\{(a,b)\in \gf_{2^n}^2: {\rm DBCT}_{F}(a,d)=4096\}|=127;\\
&|\{(a,b)\in \gf_{2^n}^2: {\rm DBCT}_{F}(a,d)=136\}|=63;\\
&|\{(a,b)\in \gf_{2^n}^2: {\rm DBCT}_{F}(a,d)=76\}|=126;\\
&|\{(a,b)\in \gf_{2^n}^2: {\rm DBCT}_{F}(a,d)=68\}|=756;\\
&|\{(a,b)\in \gf_{2^n}^2: {\rm DBCT}_{F}(a,d)=64\}|=2268;\\
&|\{(a,b)\in \gf_{2^n}^2: {\rm DBCT}_{F}(a,d)=60\}|=756.\\
\end{aligned}
\end{equation*}
The numerical results mentioned above are all consistent with the statement of Proposition \ref{proposotion-inverse-DBCT}.
\end{example}

The following result gives the double boomerang uniformity of $F(x)$ over $\gf_{2^n}$.

\begin{corollary}\label{dbu}
Let $F(x)=x^{2^n-2}$ be the inverse function over $\mathbb{F}_{2^n}$.  Then the double boomerang uniformity of $F(x)$ satisfies
$$\beta_d(F)=
\begin{cases}
    2^{n+1}, &  {\rm if}\,\,  n>1\,\, {\rm is\,\, odd}; \\
    2^{n+1}+8, &  {\rm if}\,\,  n>2\,\, {\rm is\,\, even}.\\
\end{cases}
$$
\end{corollary}

\section{Conclusion}\label{con-remarks}

In this paper, we presented the UBCT, LBCT, respectively, of the inverse function $F(x)=x^{2^n-2}$ over $\gf_{2^n}$ for arbitrary $n$. Moreover, we study the DBCT of $F(x)$ over $\gf_{2^n}$, and further determine the double boomerang uniformity of $F(x)$. Our achievements are obtained by solving specific equations over $\gf_{2^n}$ and by developing techniques to calculate the exact value of each entry on each table and determining the precise number of elements with a given entry. It is an interesting and challenging problem to completely determine the DBCT of other classes of functions over $\gf_{2^n}$, besides the ones we considered.

\section*{Acknowledgments}

This work was supported by the National Key Research and Development Program of China (No. 2021YFA1000600), the National Natural Science Foundation of China (Nos. 62072162, 12201193), the Natural Science Foundation of Hubei Province of China (No. 2021CFA079), the Application Foundation Frontier Project of Wuhan Science and Technology Bureau (No. 2020010601012189) and the Innovation Group Project of the Natural Science Foundation of Hubei Province of China (No. 2023AFA021).

\section*{Appendix A}
{\textbf{The proof of Lemma \ref{inverse-UBCT}:}}
 \begin{proof}
To proof this lemma, we need to consider the number of solutions of the equation system (\ref{UBCT}),
i.e.,
\begin{equation}\label{UBCT-0}
\begin{cases}
    x^{2^n-2}+y^{2^n-2}=c, \\
    (x+a)^{2^n-2}+(y+a)^{2^n-2}=c, \\
    (x+a)^{2^n-2}+x^{2^n-2}=b.  \\
\end{cases}
\end{equation}

If $a=0$, we have that (\ref{UBCT-0}) has solutions only if $b=0$. Then (\ref{UBCT-0}) can be reduced to $x^{2^n-2}+y^{2^n-2}=c$. Since the inverse function is a permutation, then for any $x\in \gf_{2^n}$, there exists a unique $y\in \gf_{2^n}$ such that $y^{2^n-2}=c+x^{2^n-2}$. Hence, for any $c\in \gf_{2^n}$, we get
$${\rm UBCT}_{F}(a,b,c)=2^n.$$

If $a,b\ne 0$ and $c=0$, the first equation in (\ref{UBCT-0}) reduces to $x^{2^n-2}+y^{2^n-2}=0$, which implies that $x=y$. Then we have
$${\rm UBCT}_{F}(a,b,c)={\rm DDT}_{F}(a,b).$$

Next, we assume that $a, b, c\ne 0$. Obviously, the equation system (\ref{UBCT-0}) is equivalent to
\begin{equation}\label{UBCT-1}
\begin{cases}
    x^{2^n-2}+y^{2^n-2}=c, \\
    \Delta(x)=\Delta(y), \\
    (x+a)^{2^n-2}+x^{2^n-2}=b, \\
\end{cases}
\end{equation}
where $\Delta(x)=(x+a)^{2^n-2}+x^{2^n-2}$ with $a\ne 0$.

Let $\Delta(x)=\Delta(y)=\gamma$ for some $\gamma \in \gf_{2^n}^*$.  For $\Delta(x)=(x+a)^{2^n-2}+x^{2^n-2}=\gamma$ with $a\ne 0$, if $x=0,a$, then we have $a\gamma=1$. Assume that $x\ne 0,a$, then $\Delta(x)=\gamma$ can be reduced to
\begin{equation}\label{UBCT-2}
\gamma x^2+a\gamma x+a=0,
\end{equation}
where $\gamma \in \gf_{2^n}^*$.

We start by considering the case $n$ is odd.

{\textbf{Case 1:}} Assume that $a\gamma=1$. Then we have ${\rm Tr}_{1}^n(\frac{1}{a\gamma})=1$, since $n$ is odd. By Lemma \ref{lemma1-root}, this implies that (\ref{UBCT-2}) has no solution when $a\gamma=1$. Thus $\Delta(x)=\gamma$ has two solutions $x=0,a$ when $a\gamma=1$. Similarly $\Delta(y)=\gamma$ has two solutions $y=0,a$ when $a\gamma=1$. Since $x^{2^n-2}+y^{2^n-2}=c$ with $c\ne 0$, then we have that $\Delta(x)=\Delta(y)=\gamma$ has two possible solutions $(x,y)=(0,a), (a,0)$. We then determine whether $(x,y)=(0,a), (a,0)$ are the solutions of (\ref{UBCT-0}). Substituting $(x,y)=(0,a), (a,0)$ into (\ref{UBCT-0}), we always have
\begin{equation*}
\begin{cases}
    \frac{1}{a}=c, \\
    \frac{1}{a}=b.\\
\end{cases}
\end{equation*}
This means that $(x,y)=(0,a), (a,0)$ are the two solutions of (\ref{UBCT-0}) when $b=c$ and $ab=1$.

{\textbf{Case 2:}} Assume that $a\gamma\ne 1$. By Lemma \ref{lemma1-root}, if ${\rm Tr}_{1}^n(\frac{1}{a\gamma})=0$, we know that (\ref{UBCT-2}) has two solutions over $\gf_{2^n}$, namely $\omega$ and $\omega+a$, where $\omega \ne 0,a$. Similarly $\Delta(y)=\gamma$ has two solutions $y=\omega,\omega+a$. Since $x^{2^n-2}+y^{2^n-2}=c$ with $c\ne 0$, then we have that $\Delta(x)=\Delta(y)=\gamma$ has two possible solutions $(x,y)=(\omega, \omega+a), (\omega+a,\omega)$. We then determine whether $(x,y)=(\omega, \omega+a), (\omega+a,\omega)$ are the solutions of (\ref{UBCT-0}). Substituting $(x,y)=(\omega, \omega+a), (\omega+a,\omega)$ into (\ref{UBCT-0}), we always have
\begin{equation*}
\begin{cases}
    \frac{1}{\omega+a}+\frac{1}{\omega}=c, \\
    \frac{1}{\omega+a}+\frac{1}{\omega}=b. \\
\end{cases}
\end{equation*}
which is equivalent to
\begin{equation*}
\begin{cases}
    b=c, \\
    c\omega^2+ac\omega+a=0.
\end{cases}
\end{equation*}
Since $ac\ne 0$, by Lemma \ref{lemma1-root}, the second equation of the above equation system has solutions over $\gf_{2^n}$ only if ${\rm Tr}_{1}^n(\frac{1}{ac})=0$. Therefore, we get $(x,y)=(\omega, \omega+a), (\omega+a,\omega)$  are two solutions of (\ref{UBCT-0}) if and only if $b=c$ and ${\rm Tr}_{1}^n(\frac{1}{ac})=0$.

Now we study the case when $n$ is even.

{\textbf{Case 1:}} Assume that $a\gamma=1$. Then we have ${\rm Tr}_{1}^n(\frac{1}{a\gamma})=0$, since $n$ is even. By Lemma \ref{lemma1-root}, this implies that (\ref{UBCT-2}) has two solutions, namely $\omega$ and $\omega+a$, where $\omega \ne 0,a$. Thus we have that $\Delta(x)=\gamma$ has four solutions $x=0, a, \omega, \omega+a$ when $a\gamma=1$, where $\omega \ne 0,a$. Similarly $\Delta(y)=\gamma$ has four solutions $y=0, a, \omega, \omega+a$ when $a\gamma=1$, where $\omega \ne 0,a$. Since $x^{2^n-2}+y^{2^n-2}=c$ with $c\ne 0$, then we have that $\Delta(x)=\Delta(y)=\gamma$ has twelve possible solutions
$$
\begin{aligned}
(x,y)=&(0, \omega), (\omega, 0), (a, \omega+a), (\omega+a, a),\\
&(a, \omega), (\omega, a), (0, \omega+a), (\omega+a, 0),\\
&(0, a), (a, 0), (\omega+a, \omega), (\omega, \omega+a).\\
\end{aligned}
$$
Substituting this twelve possible solutions into (\ref{UBCT-0}), by computing,  we can consider the following three subcases.

{\textbf{Subcase 1:}} Substituting $(x,y)=(0, \omega), (\omega, 0), (a, \omega+a), (\omega+a, a)$ into (\ref{UBCT-0}), we always have
\begin{equation*}
\begin{cases}
    \omega=\frac{1}{c},\\
    ab=1, \\
    ac^2+abc+b=0.
\end{cases}
\end{equation*}
Since $\omega\ne 0,a$ and $c\ne 0$, then we have $ac\ne 1$.
Thus we obtain that $(x,y)=(0, \frac{1}{c}), (\frac{1}{c}, 0), (a, \frac{1}{c}+a), (\frac{1}{c}+a, a)$ are the four solutions of (\ref{UBCT-0}) when $b\ne c$, $ab=1$ and $ac^2+abc+b=0$.

{\textbf{Subcase 2:}} Substituting $(x,y)=(a, \omega), (\omega, a), (0, \omega+a), (\omega+a, 0)$ into (\ref{UBCT-0}), we always have
\begin{equation*}
\begin{cases}
    \omega=\frac{1}{c}+a,\\
    ab=1, \\
    ac^2+abc+b=0.
\end{cases}
\end{equation*}
Since $\omega\ne 0,a$ and $c\ne 0$, then we have $ac\ne 1$.
Thus we obtain that $(x,y)=(0, \frac{1}{c}), (\frac{1}{c}, 0), (a, \frac{1}{c}+a), (\frac{1}{c}+a, a)$ are the four solutions of (\ref{UBCT-0}) when $b\ne c$, $ab=1$ and $ac^2+abc+b=0$.

{\textbf{Subcase 3:}} Substituting $(x,y)=(0, a), (a, 0), (\omega+a, \omega), (\omega, \omega+a)$ into (\ref{UBCT-0}), we always have
\begin{equation*}
\begin{cases}
    b=c, \\
    ab=1. \\
\end{cases}
\end{equation*}
According to the above discussion, we know that $(x,y)=(0, \frac{1}{c}), (\frac{1}{c}, 0), (a, \frac{1}{c}+a), (\frac{1}{c}+a, a)$ are the four solutions of (\ref{UBCT-0}) when $b\ne c$ and $ab=1$ and $ac^2+abc+b=0$, and $(x,y)=(0, a), (a, 0), (\omega+a, \omega), (\omega, \omega+a)$ are the four solutions of (\ref{UBCT-0}) when $b=c$ and $ab=1$, where $\omega \ne 0,a$.

{\textbf{Case 2:}} Assume that $a\gamma\ne 1$. Similar to the Case 2 for $n$ is odd. Here we omit this proof.
This completes the proof.
\end{proof}


\begin{thebibliography}{99}

\bibitem{v-BCT} Biryukov A., Khovratovich D.: Related-key cryptanalysis of the full AES-192 and AES-256. In Advances in Cryptology-ASIACRYPT, M. Matsui, Ed. Berlin, Germany: Springer, 5912: 1-18 (2009)

\bibitem{diff-at} Biham E., Shamir A.: Differential cryptanalysis of DES-like cryptosystems. J. Cryptology, 4(1): 3-72 (1991)


\bibitem{def-FBCT} Boukerrou H., Huynh P., Lallemand V., Mandal B., Minier M.: On the Feistel Counterpart of the Boomerang Connectivity Table: Introduction and Analysis of the FBCT, IACR Transactions on Symmetric Cryptology, 2020(1): 331-362 (2020)

\bibitem{BCT-u} Boura C., Canteaut A.: On the boomerang uniformity of cryptographic S-boxes. IACR Transactions on Symmetric Cryptology, 2018(3): 290-310 (2018)


\bibitem{Kloosterman} Carlitz L.: Kloosterman sums and finite field extensions. Acta Arithmetica, 16(2): 179-193 (1969)



\bibitem{BCT} Cid C., Huang T., Peyrin T., Sasaki Y., Song L.: Boomerang connectivity table: a new cryptanalysis tool. in: J. Nielsen, V. Rijmen (Eds.), Advances in Cryptology-EUROCRYPT 2018, Springer, Cham, 10821: 683-714 (2018)

\bibitem{ULBCT} Delaune S., Derbez P., Vavrille M.: Catching the fastest boomerangs application to SKINNY. IACR Trans. Symmetric Cryptol. 2020(4): 104-129 (2020)

\bibitem{ex-FBCT} Eddahmani S., Mesnager S.:  Explicit values of the DDT, the BCT, the FBCT, and the FBDT of the inverse, the gold, and the Bracken-Leander S-boxes. Cryptogr. Commun., 14: 1301-1344 (2022)

\bibitem{def-DBCT} Hadipour H., Bagheri N., Song L.: Improved rectangle attacks on SKINNY and CRAFT. IACR Trans. Symmetric Cryptol., 2021(2): 140-198 (2021)

\bibitem{v-BCT1} Kelsey J., Kohno T., Schneier B.: Amplified boomerang attacks against reduced-round MARS and Serpent. International Workshop on Fast Software Encryption. Berlin, Heidelberg: Springer Berlin Heidelberg, 2000: 75-93




\bibitem{def-squ} Lidl R., Niederreiter H.: Finite fields. Cambridge university press, 1997

\bibitem{def-DDT}  Nyberg K.: Differential uniform mappings for cryptography. In: Workshop on the Theory and Application of Cryptographic Techniques. Lofthus: Springer, 1993: 55-64


\bibitem{BDT'} Song L., Qin X., Hu L.: Boomerang connectivity table revisited. Application to SKINNY and AES. IACR Trans. Symmetric Cryptol. 2019(1): 118-141 (2019)



\bibitem{BA} Wagner D.: The boomerang attack. In: Knudsen L. (eds) Fast Software Encryption. FSE 1999. Lecture Notes in Computer Science. Berlin, Heidelberg, Springer, 1636: 156-170 (1999)

\bibitem{BDT} Wang H., Peyrin T.: Boomerang switch in multiple rounds. application to AES variants and Deoxys. IACR Trans. Symm. Cryptol., 2019(1): 142-169 (2019)

\bibitem{IT-D23}Wang L., Song L., Wu B., Rahman N., Isobe T.: Revisiting the boomerang attack from a perspective of 3-differential. IEEE Trans. Inf. Theory, DOI: 10.1109/TIT.2023.3324738. (2023)



\bibitem{p-DBCT} Yang Q., Song L., Sun S., Shi D., Hu L.: New properties of the double boomerang connectivity table. IACR Transactions on Symmetric Cryptology, 2022(4): 208-242 (2022)

\end{thebibliography}
\end{document}